\documentclass[a4paper,12pt]{article}
\usepackage[utf8]{inputenc}

\usepackage{amsmath,amssymb,theorem}
\usepackage[english]{babel}

\theoremstyle{plain}
\newtheorem{theorem}{Theorem}
\newtheorem{lemma}[theorem]{Lemma}
\newtheorem{definition}[theorem]{Definition}
\newtheorem{corollary}[theorem]{Corollary}

\newenvironment{proof}
{\begin{trivlist}\item[]{{\sc Proof.}}}{\hfill{$\square$}\noindent\end{trivlist}}

\begin{document}

\title{The average representation -- a cornucopia of power indices?}
\author{Serguei Kaniovski\footnote{
Austrian Institute of Economic Research (WIFO), Austria. \newline E-mail: {serguei.kaniovski@wifo.ac.at}.}
\; and \;
Sascha Kurz\footnote{
Department of Mathematics, University of Bayreuth, Germany. \newline E-mail: {sascha.kurz@uni-bayreuth.de}. (Corresponding author.)}
}
\date{\today}
\maketitle

\begin{abstract}
For the classical power indices there is a disproportion between power and relative weights, in general. We introduce 
two new indices, based on weighted representations, which are proportional to suitable relative weights and which also 
share several important properties of the classical power indices. Imposing further restrictions on the set of representations 
may lead to a whole family of such indices. 
    
  \medskip
  
\end{abstract}

\section{Introduction}
The observation, that the distribution of power in weighted majority games differs from the distribution of voting weights, has 
motivated the development of a theory of power measurement. A famous example considers three voters, having $50$, $49$, and $1$ 
votes. The motion is passed if the total number of votes in favor exceeds $50$. Since any two voters, but none alone, can pass 
the motion, any reasonable power index assigns equal power to all three voters.

The power distribution in the above example markedly differs from the relative weight distribution: 
$\left\Vert\left(\frac{50}{100},\frac{49}{100},\frac{1}{100}\right)-\left(\frac{1}{3},\frac{1}{3},\frac{1}{3}\right)\right\Vert_1
=\frac{97}{150}\approx 0.65$. One reason for this disagreement is the fact that voting weights are not unique. For three 
voters having $100$ votes in total, there are $1176$ integer valued weight distributions being extendable by a suitable 
quota to the same game.  
% consistent with the above power vector. 
An example is given by the weights $34$, $33$, $33$ and a quota of $60$. There are $13872$ possibilities to represent the game 
if the (integer) quota is considered to be part of the specification.

Having plenty of representations to choose from, can we choose voting weights that accurately reflect power measured by some 
index? The theoretical literature shows that, in general, we cannot, although there are particular cases when it may be 
possible. Recently \cite{houy2014geometry} have characterized a class of weighted majority games, which admit a 
representation using their respective Banzhaf distribution. For the nucleolus, we known that $(q(x^\star,v),x^\star)$ is a 
representation of a constant-sum weighted majority game, where $x^{\star}(v)$ denotes the nucleolus of $v$, and 
$q(x^\star,v)$ denotes the corresponding maximum excess \cite[Theorem 20.52]{game_theory_book}. If the weights are close to 
the average weight of the voters, then the nucleolus is close to the relative weight distribution; the two may even coincide 
under certain conditions, see \cite{kurz2013nucleolus}. The existing power indices are not representation compatible. One
exception is the recently introduced minimum sum representation (MSR) index \cite{freixas2014minimum}. Also the 
Colomer index \cite{colomer1995paradox} uses weights of a majority game in its specification, but the index depends on the 
given representation, instead of the underlying simple game.

In this paper, we show how to construct representation compatible power indices for weighted majority games. 

\section{Games and representations}
A simple game $v$ is a mapping $v:2^N\rightarrow \{0,1\}$, where $N=\{1,\dots,n\}$ is the set of voters, such that
$v(\emptyset)=0$, $v(N)=1$, and $v(S)\le v(T)$ for all $S\subseteq T\subseteq N$. The subsets $S\subseteq N$ are called 
coalitions of $v$. We call a coalition $S$ winning if $v(S)=1$, and losing otherwise. If $S$ is a winning coalition and 
none of its proper subsets is winning, it is called a minimal winning coalition. Similarly, if $T$ is a losing coalition 
and none of its proper supersets is losing, it is called a maximal losing coalition. A voter $i\in N$ with $v(S)=v(S\cup\{i\})$ 
for all $S\subseteq N\backslash\{i\}$ is called a dummy (or null player by some authors).

A weighted majority game is a simple game $v$, such that there exist real numbers $w_1,\dots,w_n\ge 0$ and 
$q>0$ with $\sum_{s\in S} w_s\ge q$ for all winning coalitions $S\subseteq N$ and $\sum_{s\in T} w_s< q$ for 
all losing coalitions $T\subseteq N$.  We write $v=[q;w_1,\dots,w_n]$, where we call $(q,w_1,\dots,w_n)$ a
representation of $v$. A weight vector $(w_1,\dots,w_n)$ is called feasible for $v$, if there
exists a quota $q$ such that $(q;w_1,\dots,w_n)$ is a representation of $v$. For our initial example $[2;1,1,1]$, the 
weight vector $(49,48,3)$ is feasible, while $(50,25,25)$ is not feasible.

We collect some basic well known facts about representations of weighted majority games:
\begin{lemma}
  Each weighted majority game $v$ admits a representation $(q,w_1,\dots,w_n)$ with
  $w_1,\dots,w_n\ge 0$, $q>0$, and 
  \begin{enumerate}
    \item[(1)] $\sum_{i=1}^n w_i=1$, $q\in(0,1]$;
    \item[(2)] $\sum_{i=1}^n w_i=1$, $q\in(0,1]$, and $w_i=0$ for all dummies;
    \item[(3)] $q\in\mathbb{N}$, $w_i\in \mathbb{N}$.
  \end{enumerate}
\end{lemma}

We call a representation satisfying the conditions of~(1) a normalized representation, and those satisfying
the conditions of (3) an integer representation. A representation with $w_i=0$ for all dummies $i\in N$ is called dummy 
revealing.\footnote{The problem of checking whether a voter is a dummy in a general (integer)
representation it coNP-complete \cite[Theorem 4.4]{computational_aspects}.}

Algorithmic checks and descriptions whether a given simple game is weighted have been studied extensively in the literature, see 
e.g.\ \cite{0943.91005}.

\begin{lemma}
  \label{lemma_set_of_normalized_feasible_weights}
  The set of all normalized weight vectors $w\in\mathbb{R}^n_{\ge 0}$, $\sum_{i=1}^n w_i=1$ being feasible for a given 
  weighted majority game $v$ is given by the intersection
  $\sum_{i\in S} w_i>\sum_{i\in T} w_i$ for all pairs $(S,T)$, where $S$ is a minimal winning and $T$ 
  is a maximal losing coalition of $v$.
\end{lemma}

\begin{lemma}
  \label{lemma_set_of_normalized_representations}
  The set of all normalized representations $(q,w)\in\mathbb{R}^{n+1}_{\ge 0}$, $q\in(0,1]$, $\sum_{i=1}^n w_i=1$ 
  representing a given weighted majority game $v$ is given by the intersection
  $
    \sum_{i\in S} w_i\ge q$, $\sum_{i\in T} w_i<q
  $
  for all minimal winning coalitions $S$ and all maximal losing coalitions $T$.
\end{lemma}

\section{Power indices}
Let $\mathcal{S}_n$ denote the set of simple games on $n$ voters, and $\mathcal{W}_n$ the set of weighted majority games on 
$n$ voters. A power index for $\mathcal{C}\in\{\mathcal{S}_n,\mathcal{W}_n\mid n\in\mathbb{N}\}$ is a 
mapping $g:\mathcal{C}\to\mathbb{R}^n$, where $n$ denotes the number of voters in each game of $\mathcal{C}$. Usually, 
we define a vector-valued power index by defining its elements $g_i$, the power of a voter $i\in N$. The Shapley-Shubik 
index is given by
$$
  SSI_i(v)=\sum_{S\subseteq N\backslash\{i\}} \frac{|S|!(|N|-1-|S|)!}{|N|!}\cdot (v(S\cup\{i\})-v(S)).
$$

\begin{definition}
  \label{def_des_properties}
  Let $g:\mathcal{C}\rightarrow \mathbb{R}^{|N|}=(g_i)_{i\in N}$ be a power index for $\mathcal{C}$. We say that
  \begin{enumerate}
    \item[(1)] $g$ is symmetric: if for all $v\in\mathcal{C}$ and any bijection $\tau:N\rightarrow \tau$ we have
               $g_{\tau(i)}(\tau v)=g_i(v)$, where $\tau v(S)=v(\tau(S))$ for all $S\subseteq N$;
    \item[(2)] $g$ is positive: if for all $v\in\mathcal{C}$ we have $g_i(v)\ge 0$ and $g(v)\neq 0$;
    \item[(3)] $g$ is efficient: if for all $v\in\mathcal{C}$ we have $\sum_{i=1}^n g_i(v)=1$;
    \item[(4)] $g$ satisfies the dummy property: if for all $v\in\mathcal{C}$ and all dummies $i$ of $v$ we have 
               $g_i(v)=0$.
  \end{enumerate}
\end{definition}

The Shapley-Shubik index is symmetric, positive, efficient, and satisfies the dummy property.

In this paper, we consider power indices that can be defined on the set of weighted games. Having proportionality of weights 
and power in mind we define:

\begin{definition}
  \label{def_representation_compatible}
  A power index $g:\mathcal{W}_n\rightarrow\mathbb{R}^n$ for weighted majority games on $n$ voters is called representation compatible
  if $(g_1(v),\dots,g_n(v))$ is feasible for all $v\in\mathcal{W}_n$.
\end{definition}

We remark that the Shapley-Shubik index is representation compatible for $\mathcal{W}_n$ if and only if $n\le 3$. Below,
we list the weighted majority games with up to 3 voters (in minimum sum integer representation), and the representation 
given by the Shapley-Shubik vector.
\begin{eqnarray*}
  [1;1]=\left[1;1\right] &   [1;1,0,0]=\left[\frac{6}{6};\frac{6}{6},\frac{0}{6},\frac{0}{6}\right] & 
  [2;1,1,1]=\left[\frac{4}{6};\frac{2}{6},\frac{2}{6},\frac{2}{6}\right] \\[1.5mm] 
  [1;1,0]=\left[\frac{2}{2};\frac{2}{2};\frac{0}{2}\right] & [1;1,1,0]=\left[\frac{3}{6};\frac{3}{6},\frac{3}{6},\frac{0}{6}\right] &
  [3;1,1,1]=\left[\frac{6}{6};\frac{2}{6},\frac{2}{6},\frac{2}{6}\right] \\[1.5mm]
  [1;1,1]=\left[\frac{1}{2};\frac{1}{2},\frac{1}{2}\right] & [2,1,1,0]=\left[\frac{6}{6};\frac{3}{6},\frac{3}{6},\frac{0}{6}\right] & 
  [3;2,1,1]=\left[\frac{5}{6};\frac{4}{6},\frac{1}{6},\frac{1}{6}\right] \\[1.5mm]
  [2,1,1]=\left[\frac{2}{2};\frac{1}{2},\frac{1}{2}\right] & [1;1,1,1]=\left[\frac{2}{6};\frac{2}{6},\frac{2}{6},\frac{2}{6}\right] &
  [2;2,1,1]=\left[\frac{2}{6};\frac{4}{6},\frac{1}{6},\frac{1}{6}\right]
\end{eqnarray*}

For $n\ge 4$, consider the example $v=[3;2,1,1,1,0,\dots,0]$ with $n-4$ dummies. The Shapley-Shubik index of $v$ is 
given by $\left(\frac{1}{2},\frac{1}{6},\frac{1}{6},\frac{1}{6},0,\dots,0\right)$. Since $\{2,3,4\}$ is a winning coalition with
weight $\frac{1}{2}$, and $\{1\}$ is a losing coalition with an equal weight, the Shapley-Shubik vector cannot be
feasible.

\section{Representation compatible power indices}

The aim of this paper was to design power indices, which are representation compatible. Given a set of representations
of the same weighted majority game $v$, each convex combination gives a representation of $v$ too. So a simple idea
to construct a representation compatible power index is to specify the set of representations and the weights of the
convex combination.

\subsection{The average weight index}

\begin{definition}
  The average weight index of a weighted majority game $v$ is the average of all normalized\footnote{Taking all 
  weight vectors instead of the normalized ones does not make a difference.} weight vectors which are feasible 
  for $v$.
\end{definition}

For $[3;2,1,1]$ we have already
mentioned the sets of minimal winning and maximal losing coalitions. Applying 
Lemma~\ref{lemma_set_of_normalized_feasible_weights} gives the constraints\\[-3mm]
\begin{eqnarray*}
  w_1+w_2 > w_1 &\Longleftrightarrow& w_2>0\\ 
  w_1+w_3 > w_1 &\Longleftrightarrow& w_3>0\\
  w_1+w_2 > w_2+w_3 &\Longleftrightarrow& w_1>w_3\\
  w_1+w_3 > w_2+w_3 &\Longleftrightarrow& w_1>w_2,
\end{eqnarray*} 
in addition to $w_1,w_2,w_3\ge 0$ and $w_1+w_2+w_3=1$. Eliminating the variable $w_3$ via $w_3=1-w_1-w_2$ and removing the
redundant constraints leaves
\begin{eqnarray*}
  w_2>0            &\Longleftrightarrow& w_2>0\\
  1-w_1-w_2>0      &\Longleftrightarrow& w_2<1-w_1\\
  w_1>1-w_1-w_2    &\Longleftrightarrow& w_2>1-2w_1\\
  w_1>w_2          &\Longleftrightarrow& w_2<w_1
\end{eqnarray*}
Since we need $1-2w_1<w_1$ and $1-w_1>0$, we have $w_1\in\big(\frac{1}{3},1\big)$.
For $w_1\in\big(\frac{1}{3},\frac{1}{2}\big)$ the conditions condense to $w_2\in \big(1-2w_1,w_1\big)$ and 
for $w_1\in\big[\frac{1}{2},1\big)$ the conditions condense to $w_2\in \big(0,1-w_1\big)$.

The (scaled) average weight for voter~$1$ is given by
$$
  \int_{\frac{1}{3}}^{\frac{1}{2}}\int_{1-2w_1}^{w_1} w_1\operatorname{d}w_2\,\operatorname{d}w_1+
  \int_{\frac{1}{2}}^{1}\int_{0}^{1-w_1} w_1\operatorname{d}w_2\,\operatorname{d}w_1
  =\frac{1}{54}+\frac{1}{12}=\frac{11}{108}.
$$
For voter~$2$ we similarly obtain
$$
  \int_{\frac{1}{3}}^{\frac{1}{2}}\int_{1-2w_1}^{w_1} w_2\operatorname{d}w_2\,\operatorname{d}w_1+
  \int_{\frac{1}{2}}^{1}\int_{0}^{1-w_1} w_2\operatorname{d}w_2\,\operatorname{d}w_1
  =\frac{1}{48}+\frac{5}{432}=\frac{7}{216}.
$$
Since the volume of the feasible region is given by
$$
  \int_{\frac{1}{3}}^{\frac{1}{2}}\int_{1-2w_1}^{w_1} 1\operatorname{d}w_2\,\operatorname{d}w_1+
  \int_{\frac{1}{2}}^{1}\int_{0}^{1-w_1} 1\operatorname{d}w_2\,\operatorname{d}w_1
  =\frac{1}{8}+\frac{1}{24}=\frac{1}{6},
$$
we have
$$
  \int_{\frac{1}{3}}^{\frac{1}{2}}\int_{1-2w_1}^{w_1} w_3\operatorname{d}w_2\,\operatorname{d}w_1+
  \int_{\frac{1}{2}}^{1}\int_{0}^{1-w_1} w_3\operatorname{d}w_2\,\operatorname{d}w_1
  =\frac{1}{6}-\frac{11}{108}-\frac{7}{216}=\frac{7}{216}.
$$
Normalizing, or dividing by the volume of the feasible region, yields the power distribution
$
  \left(\frac{11}{18},\frac{7}{36},\frac{7}{36}\right),
$
with a norm-$1$-distance of $\frac{1}{9}$ to the respective Shapley-Shubik vector $\left(\frac{2}{3},\frac{1}{6},\frac{1}{6}\right)$.

Actually we have dealt with the polyhedron 
$P=\big\{(w_1,w_2)\in \mathbb{R}^2\mid w_2\ge 0, w_2\le 1-w_1, w_2\ge 1-2w_1, w_2\le w_1\big\}$,
i.e., we have replaced strict inequalities by the corresponding non-strict inequalities, 
and considered the integrals $\int_p w_1\,\operatorname{d}w_1$, $\int_p w_2\,\operatorname{d}w$, and 
$\int_p 1\,\operatorname{d}w$. This modification is permitted since in general the polyhedron $P$ (after the elimination of 
variable $w_n$) is full dimensional, i.e., it has dimension $n-1$, so that the subspaces where equality hold in of of 
the inequalities have volume zero.

\begin{lemma}
  \label{lemma_full_dimension}
  For each weighted majority game $v$ there exist positive real numbers $\tilde{q},\tilde{w}_1,\dots,\tilde{w}_{n-1}$ 
  and a parameter $\alpha>0$ such that 
  \begin{equation}
    \left(\tilde{q}+\delta_0,\tilde{w}_1+\delta_1,\dots,\tilde{w}_{n-1}+\delta_{n-1},1-\sum_{i=1}^{n-1}\left(\tilde{w}_i+\delta_i\right)\right)
  \end{equation}
  is a normalized representation of $v$ for all $\delta_i\in[-\alpha,\alpha]$, $0\le i\le n-1$.  
\end{lemma}
\begin{proof}
  Let $(q,w_1,\dots,w_n)$ be an integer representation of $v$, i.e.\ the weight of each winning coalition is at least $q$, and 
  the weight of each losing coalition is at most $q-1$. Since $\big((n+1)q,(n+1)w_1+1,\dots,(n+1)w_n\big)$ is also
  an integer representation of $v$, we additionally assume w.l.o.g.\ that $w_i\ge 1$ for all $1\le i\le n$. One can easily check
  that also
  $
    \left(q-\frac{2}{5}+\tilde{\delta}_0,w_1+\tilde{\delta}_1,\dots,w_n+\tilde{\delta}_n\right)
  $
  is a representation of $v$ for all $\tilde{\delta}_i\in\left[-\frac{1}{5n},\frac{1}{5}{n}\right]$, $0\le i\le n$. 
  With $s=\sum_{i=1}^n w_i$ let $\tilde{q}=\frac{1}{s}\cdot\left(q-\frac{2}{5}\right)$ and $\tilde{w}_i=\frac{1}{s}\cdot w_i$ 
  for all $1\le i\le n-1$. The choice of a suitable $\alpha$ is a bit technical ($\alpha=\frac{1}{10ns}$ does work), but its existence 
  is guaranteed from our construction.
\end{proof}

\begin{definition}
  For each weighted majority game $v$ the (normalized) weight polyhedron $P^{\text{weight}}(v)$ is given by $P^{\text{weight}}(v)=
  \big\{w\in\mathbb{R}^n_{\ge 0}\mid \Vert w\Vert_1=1,\,w(S)\ge w(T)\forall\text{ min.}$ $\text{winning }S
    \text{ and all max.~losing }T\big\}
  $.
\end{definition}

By fixing the quota at a suitable value we can directly conclude from Lemma~\ref{lemma_full_dimension}:

\begin{corollary}
  The $n-1$-dimensional volume of $P^{\text{weight}}(v)$ is non-zero for each weighted majority game $v$.  
\end{corollary}

\begin{lemma}
  \label{lemma_compute_average_weight_index}
  The average weight index of a weighted majority game $v$ is given by
  \begin{equation}
    \frac{1}{\int_{P^{\text{weight}}(v)}\operatorname{d}w}\cdot\left(
    \int_{P^{\text{weight}}(v)}w_1\operatorname{d}w,\dots,\int_{P^{\text{weight}}(v)}w_n\operatorname{d}w
    \right). 
  \end{equation}
\end{lemma}

We remark that the average weight representations is the center of mass of the polyhedron $P^{\text{weight}}(v)$.

\subsection{The average representation index}

As mentioned already in the introduction, one may consider the quota as being part of the weighted representation. 
To this end we introduce:

\begin{definition}
  The average representation index of a weighted majority game $v$ is the average of all normalized\footnote{Taking all 
  representations instead of the normalized ones does not make a difference.} representations of $v$.
\end{definition}

\begin{definition}
  For each weighted majority game $v$ the (normalized) representation polyhedron $P^{\text{rep}}(v)$ is given by
  $P^{\text{rep}}(v)=\Big\{(q,w)\in\mathbb{R}^{n+1}_{\ge 0}\mid \sum_{i=1}^n w_i=1,\,q\le 1,\,w(S)\ge q \,
  \forall\text{ min.~winning coalitions }S,
   w(T)\le q \,\forall \text{ max.~losing coalitions }T\Big\}$.
\end{definition}

Using Lemma~\ref{lemma_set_of_normalized_representations} and Lemma~\ref{lemma_full_dimension} we conclude:

\begin{lemma}
  \label{lemma_compute_average_representation_index}
  The average representation index of a weighted majority game $v$ is given by
  \begin{equation}
    \frac{1}{\int_{P^{\text{rep}}(v)}\operatorname{d}(q,w)}\cdot\left(
    \int_{P^{\text{rep}}(v)}w_1\operatorname{d}(q,w),\dots,\int_{P^{\text{rep}}(v)}w_n\operatorname{d}(q,w)
    \right). 
  \end{equation}
\end{lemma}

For our example $v=[3;2,1,1]$ we have
$$
  P^{\text{rep}}(v)=\big\{(q,w)\in\mathbb{R}^4_{\ge 0}\mid \sum\limits_{i=1}^3 w_i=1, w_1+w_2\ge q, w_1+w_3\ge q,w_1\le q,w_2+w_3\le q\big\}, 
$$
and
\begin{eqnarray*}
  \int\limits_{P^{\text{rep}}(v)}\operatorname{d}(q,w) &=& \int\limits_{\frac{1}{2}}^{\frac{2}{3}}\int\limits_{1-q}^{q}\int\limits_{q-w_1}^{1-q}
  \,\operatorname{d}w_2\operatorname{d}w_1\operatorname{d}q+\int\limits_{\frac{2}{3}}^{1}\int\limits_{2q-1}^{q}\int\limits_{q-w_1}^{1-q}
  \,\operatorname{d}w_2\operatorname{d}w_1\operatorname{d}q\\
  &=& \frac{5}{648}+\frac{1}{162}=\frac{1}{72},\\
  \int\limits_{P^{\text{rep}}(v)}w_1\operatorname{d}(q,w) &=& \int\limits_{\frac{1}{2}}^{\frac{2}{3}}\int\limits_{1-q}^{q}\int\limits_{q-w_1}^{1-q}w_1
  \,\operatorname{d}w_2\operatorname{d}w_1\operatorname{d}q
  +\int\limits_{\frac{2}{3}}^{1}\int\limits_{2q-1}^{q}\int\limits_{q-w_1}^{1-q}w_1
  \,\operatorname{d}w_2\operatorname{d}w_1\operatorname{d}q\\
  &=& \frac{31}{7776}+\frac{1}{243}=\frac{7}{864},\\
  \int\limits_{P^{\text{rep}}(v)}w_2\operatorname{d}(q,w) &=& \int\limits_{\frac{1}{2}}^{\frac{2}{3}}\int\limits_{1-q}^{q}\int\limits_{q-w_1}^{1-q}w_2
  \,\operatorname{d}w_2\operatorname{d}w_1\operatorname{d}q+\int\limits_{\frac{2}{3}}^{1}\int\limits_{2q-1}^{q}\int\limits_{q-w_1}^{1-q}w_2
  \,\operatorname{d}w_2\operatorname{d}w_1\operatorname{d}q\\
  &=& \frac{29}{15552}+\frac{1}{972}=\frac{5}{1728},
\end{eqnarray*}
so that the average representation index is given by $\left(\frac{7}{12},\frac{5}{24},\frac{5}{24}\right)$.

\subsection{Properties of the new indices}
The two newly introduced indices share several of the properties commonly required for a power index. Three of four properties 
in Definition~\ref{def_des_properties} are satisfied.

\begin{lemma}
  \label{lemma_prop_1}
  The average weight and the average representation index are symmetric, positive, and efficient, satisfies 
  strong monotonicity, but do not satisfy the dummy property. 
\end{lemma} 
\begin{proof}
  Symmetry, positivity, and efficiency are inherent in the definition of both indices. The violation of the 
  dummy property can e.g.\ be seen at example $[1;1,0]$. 
\end{proof}

The later shortcoming can be repaired using a quite general approach.

\begin{lemma}
  Given a sequence of power indices $g^n:\mathcal{C}_n\rightarrow\mathbb{R}^n$ for all $n\in\mathbb{N}$, let
  $\tilde{g}^n:\mathcal{C}_n\rightarrow\mathbb{R}^n$ be defined via $\tilde{g}^n_i(v)=g^m_i(v')$ for all non-dummies $i$ 
  and by $\tilde{g}^n_j(v)=0$ for all dummies $j$, where $m$ is the number of non-dummies in $v$ 
  and $v'$ arises from $v$ by dropping the dummies\footnote{Given a weighted majority game $v:2^N\rightarrow\{0,1\}$ with 
  $S=\{i\in N\mid i \text{ is dummy}\}$, we define the dummy reduced game $v':2^{N\backslash S}\rightarrow \{0,1\}$ via 
  $v'(T)=v(T)$ for all $T\subseteq N\backslash S$.} All $\tilde{g}^n$ satisfy the dummy property. 
\end{lemma} 

We call $\tilde{g}^n$ the dummy revealing version of a given sequence of power indices $g^n$. For the computation 
of the dummy revealing version we just have to compute the dummy reduced game $v'$ and its corresponding power distribution.

\subsection{Algorithmic computation of the new indices}

The computations from Lemma~\ref{lemma_compute_average_weight_index} and Lemma~\ref{lemma_compute_average_representation_index} can 
easily be performed using the software package \texttt{LattE} \cite{latte_integral}, i.e., there is no need to perform the nasty 
case differentiations and evaluations of multi integrals, as done for our example, by hand. 

\begin{table}[htp]
  \begin{center}
    %%\footnotesize
    \tiny
    \begin{tabular}{cccccc}
      game & av.~weight & av.~rep.\ & game & av.~weight & av.~rep.\ \\
      \hline
      $[1;1]$ & $(1)$ & $(1)$ & $[3;2,1,1,0]$ & $\left(\frac{67}{120},\frac{47}{240},\frac{47}{240},\frac{1}{20}\right)$ & $\left(\frac{41}{75},\frac{31}{150},\frac{31}{150},\frac{1}{25}\right)$ \\[1mm]
      $[1;1,0]$ & $\left(\frac{3}{4},\frac{1}{4}\right)$ & $\left(\frac{5}{6},\frac{1}{6}\right)$ & $[1;1,1,1,1]$ & $\left(\frac{1}{4},\frac{1}{4},\frac{1}{4},\frac{1}{4}\right)$ & $\left(\frac{1}{4},\frac{1}{4},\frac{1}{4},\frac{1}{4}\right)$ \\[1mm]
      $[1;1,1]$ & $\left(\frac{1}{2},\frac{1}{2}\right)$ & $\left(\frac{1}{2},\frac{1}{2}\right)$ & $[2;1,1,1,1]$ & $\left(\frac{1}{4},\frac{1}{4},\frac{1}{4},\frac{1}{4}\right)$ & $\left(\frac{1}{4},\frac{1}{4},\frac{1}{4},\frac{1}{4}\right)$ \\[1mm]
      $[2;1,1]$ & $\left(\frac{1}{2},\frac{1}{2}\right)$ & $\left(\frac{1}{2},\frac{1}{2}\right)$ & $[3;1,1,1,1]$ & $\left(\frac{1}{4},\frac{1}{4},\frac{1}{4},\frac{1}{4}\right)$ & $\left(\frac{1}{4},\frac{1}{4},\frac{1}{4},\frac{1}{4}\right)$ \\[1mm]
      $[1;1,0,0]$ & $\left(\frac{2}{3},\frac{1}{6},\frac{1}{6}\right)$ & $\left(\frac{3}{4},\frac{1}{8},\frac{1}{8}\right)$ & $[4;1,1,1,1]$ & $\left(\frac{1}{4},\frac{1}{4},\frac{1}{4},\frac{1}{4}\right)$ & $\left(\frac{1}{4},\frac{1}{4},\frac{1}{4},\frac{1}{4}\right)$ \\[1mm]
      $[1;1,1,0]$ & $\left(\frac{4}{9},\frac{4}{9},\frac{1}{9}\right)$ & $\left(\frac{11}{24},\frac{11}{24},\frac{1}{12}\right)$ & $[4;2,1,1,1]$ & $\left(\frac{23}{48},\frac{25}{144},\frac{25}{144},\frac{25}{144}\right)$ & $\left(\frac{139}{300},\frac{161}{900},\frac{161}{900},\frac{161}{900}\right)$ \\[1mm]
      $[2;1,1,0]$ & $\left(\frac{4}{9},\frac{4}{9},\frac{1}{9}\right)$ & $\left(\frac{11}{24},\frac{11}{24},\frac{1}{12}\right)$ & $[3;2,1,1,1]$ & $\left(\frac{7}{16},\frac{3}{16},\frac{3}{16},\frac{3}{16}\right)$ & $\left(\frac{43}{100},\frac{19}{100},\frac{19}{100},\frac{19}{100}\right)$ \\[1mm]
      $[1;1,1,1]$ & $\left(\frac{1}{3},\frac{1}{3},\frac{1}{3}\right)$ & $\left(\frac{1}{3},\frac{1}{3},\frac{1}{3}\right)$ & $[2;2,1,1,1]$ & $\left(\frac{23}{48},\frac{25}{144},\frac{25}{144},\frac{25}{144}\right)$ & $\left(\frac{139}{300},\frac{161}{900},\frac{161}{900},\frac{161}{900}\right)$ \\[1mm]
      $[2;1,1,1]$ & $\left(\frac{1}{3},\frac{1}{3},\frac{1}{3}\right)$ & $\left(\frac{1}{3},\frac{1}{3},\frac{1}{3}\right)$ & $[3;2,2,1,1]$ & $\left(\frac{83}{240},\frac{83}{240},\frac{37}{240},\frac{37}{240}\right)$ & $\left(\frac{103}{300},\frac{103}{300},\frac{47}{300},\frac{47}{300}\right)$ \\[1mm]
      $[3;1,1,1]$ & $\left(\frac{1}{3},\frac{1}{3},\frac{1}{3}\right)$ & $\left(\frac{1}{3},\frac{1}{3},\frac{1}{3}\right)$ & $[4;2,2,1,1]$ & $\left(\frac{83}{240},\frac{83}{240},\frac{37}{240},\frac{37}{240}\right)$ & $\left(\frac{103}{300},\frac{103}{300},\frac{47}{300},\frac{47}{300}\right)$ \\[1mm]
      $[2;2,1,1]$ & $\left(\frac{11}{18},\frac{7}{36},\frac{7}{36}\right)$ & $\left(\frac{7}{12},\frac{5}{24},\frac{5}{24}\right)$ & $[5;2,2,1,1]$ & $\left(\frac{19}{48},\frac{19}{48},\frac{5}{48},\frac{5}{48}\right)$ & $\left(\frac{23}{60},\frac{23}{60},\frac{7}{60},\frac{7}{60}\right)$ \\[1mm]
      $[3;2,1,1]$ & $\left(\frac{11}{18},\frac{7}{36},\frac{7}{36}\right)$ & $\left(\frac{7}{12},\frac{5}{24},\frac{5}{24}\right)$ & $[2;2,2,1,1]$ & $\left(\frac{19}{48},\frac{19}{48},\frac{5}{48},\frac{5}{48}\right)$ & $\left(\frac{23}{60},\frac{23}{60},\frac{7}{60},\frac{7}{60}\right)$ \\[1mm]
      $[1;1,0,0,0]$ & $\left(\frac{5}{8},\frac{1}{8},\frac{1}{8},\frac{1}{8}\right)$ & $\left(\frac{7}{10},\frac{1}{10},\frac{1}{10},\frac{1}{10}\right)$ & $[4;3,1,1,1]$ & $\left(\frac{3}{5},\frac{2}{15},\frac{2}{15},\frac{2}{15}\right)$ & $\left(\frac{29}{50},\frac{7}{50},\frac{7}{50},\frac{7}{50}\right)$ \\[1mm]
      $[1;1,1,0,0]$ & $\left(\frac{5}{12},\frac{5}{12},\frac{1}{12},\frac{1}{12}\right)$ & $\left(\frac{13}{30},\frac{13}{30},\frac{1}{15},\frac{1}{15}\right)$ & $[3;3,1,1,1]$ & $\left(\frac{3}{5},\frac{2}{15},\frac{2}{15},\frac{2}{15}\right)$ & $\left(\frac{29}{50},\frac{7}{50},\frac{7}{50},\frac{7}{50}\right)$\\[1mm]
      $[2;1,1,0,0]$ & $\left(\frac{5}{12},\frac{5}{12},\frac{1}{12},\frac{1}{12}\right)$ & $\left(\frac{13}{30},\frac{13}{30},\frac{1}{15},\frac{1}{15}\right)$ & $[3;3,2,1,1]$ & $\left(\frac{449}{840},\frac{227}{840},\frac{41}{420},\frac{41}{420}\right)$ & $\left(\frac{77}{150},\frac{41}{150},\frac{8}{75},\frac{8}{75}\right)$ \\[1mm]
      $[1;1,1,1,0]$ & $\left(\frac{5}{16},\frac{5}{16},\frac{5}{16},\frac{1}{16}\right)$ & $\left(\frac{19}{60},\frac{19}{60},\frac{19}{60},\frac{1}{20}\right)$ & $[5;3,2,1,1]$ & $\left(\frac{449}{840},\frac{227}{840},\frac{41}{420},\frac{41}{420}\right)$ & $\left(\frac{77}{150},\frac{41}{150},\frac{8}{75},\frac{8}{75}\right)$\\[1mm]
      $[2;1,1,1,0]$ & $\left(\frac{5}{16},\frac{5}{16},\frac{5}{16},\frac{1}{16}\right)$ & $\left(\frac{19}{60},\frac{19}{60},\frac{19}{60},\frac{1}{20}\right)$ & $[4;3,2,2,1]$ & $\left(\frac{193}{480},\frac{31}{120},\frac{31}{120},\frac{13}{160}\right)$ & $\left(\frac{119}{300},\frac{77}{300},\frac{77}{300},\frac{9}{100}\right)$ \\[1mm]
      $[3;1,1,1,0]$ & $\left(\frac{5}{16},\frac{5}{16},\frac{5}{16},\frac{1}{16}\right)$ & $\left(\frac{19}{60},\frac{19}{60},\frac{19}{60},\frac{1}{20}\right)$ & $[5;3,2,2,1]$ & $\left(\frac{193}{480},\frac{31}{120},\frac{31}{120},\frac{13}{160}\right)$ & $\left(\frac{119}{300},\frac{77}{300},\frac{77}{300},\frac{9}{100}\right)$ \\[1mm]
      $[2;2,1,1,0]$ & $\left(\frac{67}{120},\frac{47}{240},\frac{47}{240},\frac{1}{20}\right)$ & $\left(\frac{41}{75},\frac{31}{150},\frac{31}{150},\frac{1}{25}\right)$ 
    \end{tabular}
    \caption{The average weight and average representation index for small weighted majority games.}
    \label{table_average_indices}
  \end{center}
\end{table}

In Table~\ref{table_average_indices} we list the average weight and the average representation index for all weighted 
majority games with up to $4$~voters. We observe that the so-called dual games obtain the same average weight and average 
representation index, which can indeed be proved easily.

\section{Conclusion and future research}

We have shown how to construct power indices that respect proportionality between power and weight from average representations
of a game. By restricting the polyhedron implied by the set of minimal winning and maximal losing coalitions, we can obtain a 
representation that is also dummy revealing. This might be a hint that they principally could be suited to serve as a measurement 
for voting power in a certain context.
We do not claim that this indeed the case, but propose the challenge to eventually disqualify such a usage rigorously instead.
Having the ongoing search for the \textit{right index} in mind, we want to encourage (even) more research with respect to the 
question of indispensable properties of measurements of power.

The average representations themselves may have other uses too. They conveniently summarize the set of admissible 
representations of a weighted majority game into a unique representation, which can then be compared to power distributions 
of the classical power indices. We have already restricted the set of representations in order to obtain a dummy 
revealing index, but there still may be much more meaningful ways to introduce further restrictions. So our approach 
might lead to a cornucopia of power indices. 

We conclude the paper with a remark on integer weights. A normalization of voting weights is unreasonable if they are to 
represent the number of shares of a corporation, or the number of members of a political party. In these cases, we require the 
weights to be integers. However, there is still an interpretation of our indices in these cases, as the following convergence 
result suggests. 
 
Let us return to the initial example in the introduction, and consider the weighted majority game $v=[2;1,1,1]$. We have said 
that there are 1176 integer weight vectors being feasible for $v$ with a sum of weights 100. If we average those 
representation we obtain $\left(\frac{100}{3},\frac{100}{3},\frac{100}{3}\right)$ or a relative distribution of 
$\left(\frac{1}{3},\frac{1}{3},\frac{1}{3}\right)$, which is no surprise due to the inherent symmetry. Things get a bit more 
interesting if one considers the weighted majority game $v=[3;2,1,1]$. For a weight sum of 100, we have 1601 different weight 
vectors, and the averaged relative weight distribution is given by $(0.608832,0.195584,0.195584)$. 
For a weight sum of 1000, we obtain 166001 different weight vectors and $(0.610888,0.194556,0.194556)$.
For a weight sum of 10000, we obtain 16660001 different weight vectors and $(0.611089,0.194456,0.194456)$.
For a weight sum of 100000 we obtain 1666600001 different weight vectors and $(0.611109 0.194446 0.194446)$. The averaged 
relative weight distribution seems to converge to $\left(\frac{11}{18},\frac{7}{36},\frac{7}{36}\right)$, which is indeed the 
average weight index. This can be rigorously proven by numerically approximating the integrals of the definition 
of the average weight index using grid points only, and considering the limit of finer and finer equally distributed grids. The 
same is true if an integer valued quota is taken into account. In the limit, we would end up with the average representation index. 
Also the dummy revealing property can be transfered in this sense.

%%\bibliographystyle{authordate1}
%%\bibliography{more_indices}

\end{document}